\documentclass[11pt,letter]{article}
\usepackage[hidelinks]{hyperref} 
\usepackage{fullpage}
\usepackage[utf8x]{inputenc}
\usepackage{amsthm}
\usepackage{todonotes,wrapfig,float,graphicx,amssymb,textcomp,array,amsmath}
\usepackage{enumerate,enumitem}
\usepackage{multirow}
\usepackage{tabularx}
\usepackage{color,xcolor}
\usepackage{enumitem}
\usepackage{float}
\usepackage{nicefrac}

\definecolor{mycolor}{rgb}{0, 0, 0}
\usepackage[titletoc,title]{appendix}

\usepackage{lineno}

\title{Minimum Plane Bichromatic Spanning Trees}

\author{
Hugo A. Akitaya\thanks{Miner School of Computer \& Information Sciences, University of Massachusetts, Lowell, USA, \texttt{hugo\_akitaya@uml.edu}. Research supported by the NSF award CCF-2348067.} 
\and Ahmad Biniaz\thanks{School of Computer Science, University of Windsor, Canada, \texttt{abiniaz@uwindsor.ca}. Research supported by NSERC.}
\and Erik D. Demaine\thanks{Computer Science and Artificial Intelligence Lab, Massachusetts Institute of Technology, Cambridge,
MA, USA, \texttt{edemaine@mit.edu}.} 
\and  Linda Kleist\thanks{Institute of Computer Science, Universität Potsdam, Potsdam, Germany, \texttt{kleist@cs.uni-potsdam.de}.}
\and Frederick Stock\thanks{Miner School of Computer \& Information Sciences, University of Massachusetts, Lowell, USA, \texttt{frederick\_stock@student.uml.edu}.}
\and Csaba D. T\'{o}th\thanks{Department of Mathematics, California State University Northridge, Los Angeles, USA; and
Department of Computer Science, Tufts University, Medford, MA, USA, \texttt{csaba.toth@csun.edu}. Research supported in part by the NSF award DMS-2154347.} 
 }

\date{}
\newtheorem{lemma}{Lemma}

\newtheorem{proposition}{Proposition}

\newtheorem{theorem}{Theorem}

\newtheorem*{problem*}{Problem}
\newtheorem*{claim*}{Claim}
\newtheorem*{invariant*}{Invariant}

\definecolor{mycolor}{rgb}{0, 0, 0}
\newcommand{\revised}[1]{{\color{black} #1}}
\newcommand{\etal}{{et~al.}}
\newcommand{\OPT}{\mathrm{OPT}}
\newcommand{\MST}{{MST}}
\newcommand{\MBST}{{MinBST}}
\newcommand{\MPBST}{MinPBST}
\DeclareMathOperator{\EX}{\mathbb{E}}
\DeclareMathOperator{\PR}{\mathrm{Pr}}
\newcommand{\spath}[2]{\delta(#1,#2)}

\begin{document}
\maketitle
\begin{abstract}
For a set of red and blue points in the plane, a {\em minimum bichromatic spanning tree} (\MBST{}) is a shortest spanning tree of the points such that every edge has a red and a blue endpoint. A \MBST{} can be computed in $O(n\log n)$ time where $n$ is the number of points. In contrast to the standard Euclidean MST, which is always {\em plane} (noncrossing), a \MBST{} may have edges that cross each other. 
However, we prove that a \MBST{} is 
{\color{mycolor}quasi-plane,}
that is, it does not contain three pairwise crossing edges{\color{mycolor}, and we determine the maximum number of crossings.}

Moreover, {\color{mycolor}we study the problem of finding}
a {\em minimum plane bichromatic spanning tree} (\MPBST{}) which is a shortest bichromatic spanning tree with pairwise noncrossing edges. This problem is known to be NP-hard. The previous best approximation algorithm, due to Borgelt~\etal~(2009), has a ratio of $O(\sqrt{n})$. It is also known that 
the optimum solution can be computed
in polynomial time in some special cases, for instance, when the points are in convex position, collinear, semi-collinear,
or when one color class has constant size. We present an $O(\log n)$-factor approximation algorithm for the general case. 

\end{abstract}

\section{Introduction}

Computing a minimum spanning tree (\MST) in a graph is a well-studied problem. There exist many algorithms for this problem, among which one can mention the celebrated Kruskal's algorithm \cite{Kruskal}, Prim's algorithm \cite{Prim}, and Bor\r{u}vka's algorithm \cite{Boruvka}. The running time of these algorithms depends on the number of vertices and  edges of the input graph. For geometric graphs, where the vertices are points in the plane,
\revised{their running time depends} only on the number of vertices.

For a set $S$ of $n$ points in the plane a Euclidean MST (i.e., an 
MST of the complete 
graph on $S$ \revised{with straight-line edges and Euclidean edge weights}) can be computed in $O(n\log n)$ time. When the points of $S$ are colored by two colors, say red and blue, and every edge is required to have a red and a blue endpoint, then a spanning tree is referred to as a \emph{bichromatic spanning tree}.
A {\em minimum bichromatic spanning tree} (\MBST{}) is a bichromatic spanning tree of minimum total edge length. A \MBST{} on $S$ can be computed in $O(n\log n)$ time \cite{Biniaz2018}. When the points are collinear (all lie on a straight line) and are given in sorted order along the line this problem can be solved in linear time \cite{BandyapadhyayBB21}.

We say that two line segments {\em cross} if they share an interior point; this configuration is called a {\em crossing}. A tree is called {\em plane} if its edges are pairwise noncrossing.
The standard Euclidean \MST{} is always plane. This property is ensured by the triangle inequality, because the tree can be made shorter by replacing any two crossing edges with two noncrossing edges. 
The noncrossing property does not necessarily hold for a \MBST{}, see Figure~\ref{MinBST-fig} for an example. Two crossing edges in this example cannot be replaced with two noncrossing edges because, otherwise, we would either introduce monochromatic edges (that connect points of the same color) or disconnect the tree into two components. 

\begin{figure}[ht!]
	\centering
\includegraphics[width=.3\columnwidth]{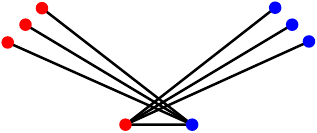}
	\caption{A bicolored point set and its  minimum bichromatic spanning tree (\MBST{}).}
	\label{MinBST-fig}
\end{figure}

Edge crossings in geometric graphs are usually undesirable as they could lead to unwanted situations such as collisions in motion planning, inconsistency in VLSI layout, and interference in wireless networks. They are also undesirable in the context of graph drawing and network visualization.
Therefore, it is natural to ask for a \emph{minimum plane bichromatic spanning tree} (\MPBST{}), a bichromatic spanning tree that is noncrossing and has minimum total edge length. {\color{mycolor}  Borgelt~\etal~\cite{Borgelt2009} proved that the problem of finding a \MPBST{} is NP-hard. They also present a polynomial-time approximation algorithm with approximation factor $O(\sqrt{n})$.}

{\color{mycolor}
In this paper we study the \MBST{} and \MPBST{} problems from  combinatorial and computational points of view. First we present an approximation algorithm, with a better factor
, for the \MPBST{} problem.} Then we prove some interesting structural properties of the \MBST{}. 

\subsection{Related work}

Problems related to bichromatic objects (such as points and lines) have been actively studied in computational geometry, for instance, the problems related to bichromatic intersection \cite{Agarwal1990,Chan10,Chan2011,Mairson1988}, bichromatic separation \cite{Alegria2023,Armaselu2019,Arora2004,Bespamyatnikh2000,Demaine2005}, and noncrossing bichromatic connection \cite{Abellanas1999,Abu-Affash2021,Biniaz2019,Biniaz2018a,Borgelt2009,Hoffmann2014,Kaneko1998,Kano2013a}. We refer the interested reader to the survey by Kaneko and Kano~\cite{Kaneko2003}.

The $O(\sqrt{n})$-approximation algorithm of  Borgelt~\etal~\cite{Borgelt2009} for the  \MPBST{} problem  lays a $(\sqrt{n}\times\sqrt{n})$-grid over the points, then  identifies a subset of grid cells as {\em core} regions and computes their Voronoi diagram, then builds a tree inside each Voronoi cell, and finally combines the trees. 

{\color{mycolor}
Let $\rho_n$ be the supremum ratio of the length of \MPBST{} to the length of \MBST{} over all sets of $n$ bichromatic points. 
Grantson~\etal~\cite{Grantson2005} show that $3/2\leq \rho_n\le n$ for all $n\geq 4$; and ask whether the upper bound can be improved. It is easily seen from the algorithm of Borgelt~et~al.~\cite{Borgelt2009} that $\rho_n\le O(\sqrt{n})$ because the planarity of the optimal solution is not used in the analysis of the approximation ratio---indeed the analysis would work even with respect to the \MBST{}.
}

Some special cases of the \MPBST{} problem can be solved to optimality in polynomial time. For instance, the problem can be solved in $O(n^2)$ time when points are collinear \cite{BandyapadhyayBB21}, in $O(n^3)$ time when points are in convex position \cite{Borgelt2009}, in $O(n^5)$ time when points are semi-collinear (points in one color class are on a line and all other points are on one side of the line) \cite{Biniaz2019}, and in $n^{O(k^5)}$ time when one color class has $k$ points for some constant $k$ \cite{Borgelt2009}.

One might wonder if a greedy strategy could achieve a better approximation ratio. A modified version of Kruskal's algorithm, that 
{\color{mycolor}successively adds a}
shortest bichromatic edge that 
{\color{mycolor}creates neither a cycle nor a crossing,}
is referred to as the {\em greedy algorithm} \cite{Borgelt2009,Grantson2005}. This algorithm, as noted in \cite[Figure~1]{Borgelt2009}, 
{\color{mycolor}does not always return a planar bichromatic tree (it does not always terminate:} 
there may be a point of one color that cannot see 
any point of the opposite color).


Abu-Affash~\etal~\cite{Abu-Affash2021} studied the {\em bottleneck} version of the plane bichromatic spanning tree problem where the goal is to minimize the length of the longest edge. They prove that this problem is NP-hard, and present an $8\sqrt{2}$-approximation algorithm.  

\subsection{Quasi-planarity}

Quasi-planarity is a measure of the proximity of an (abstract or geometric) graph to planarity. For an integer $k\ge 2$, a graph is called \emph{$k$-quasi-planar} if it can be drawn in the plane such that no $k$ edges pairwise cross. 
By this definition, a planar graph is 2-quasi-planar. A 3-quasi-planar graph is also called {\em quasi-planar}. Problems on $k$-quasi-planarity are closely related to Tur\'{a}n-type problems on the intersection graph of line segments in the plane~\cite{Agarwal1997,Angelini2020,Capoyleas1992,Fox2022}. They are also related to the size of {\em crossing families}  (pairwise crossing edges) determined by points in the plane \cite{Aronov1994,Pach2021}. Perhaps a most notable question on quasi-planarity is a conjecture by Pach, Shahrokhi, and Szegedy~\cite{Pach1996} that for any \revised{fixed} integer $k\ge 3$, \revised{there exists a constant $c_k$ such that} every $n$-vertex $k$-quasi-planar graph has \revised{at most $c_kn$} edges.
This conjecture has been verified for $k=3$ \cite{Agarwal1997} and $k=4$  \cite{Ackerman2009}.
{\color{mycolor}

A drawing of a graph is called \emph{$k$-quasi-plane} if no $k$ edges in the drawing pairwise cross, and a drawing is \emph{quasi-plane} if it is 3-quasi-plane. For example, the drawing of a tree in Figure~\ref{MinBST-fig} is quasi-plane. This concept plays an important role in decompositions of geometric graphs: Aichholzer et al.~\cite{AichholzerOOPSS22} showed recently that the complete geometric graph on $2n$ points in the plane can always be decomposed into $n$ quasi-plane spanning trees (but not necessarily into $n$ plane spanning trees). 
}

\subsection{Our contributions} 

In Section~\ref{approx-section} we present a randomized approximation algorithm with factor $O(\log n)$ for the \MPBST{} problem. Our algorithm computes a randomly shifted quadtree on the points, and then builds a planar bichromatic tree in a bottom-up fashion from the leaves of the quadtree towards the root. {\color{mycolor}We then derandomize the algorithm by discretizing the random shifts.
Our weight analysis shows that $|\text{\MBST}(S)|\leq |\text{\MPBST}(S)|\leq O(\log n)\cdot |\text{\MBST}(S)|$ for every set $S$ of $n$ bichromatic points, which implies that $\rho_n = O(\log n)$.}

In Section~\ref{quasi-section} we prove that every \MBST{} is  quasi-plane, {\color{mycolor}i.e., no three edges pairwise cross in its inherited drawing (determined by the point set). In a sense,} this means that \MBST{} is not far from 
{\color{mycolor}plane graphs. In Section~\ref{crossings-section} we determine the maximum number of crossings in a \MBST{}}.
We conclude with a list of open problems in Section~\ref{sec:conclusion}.

\section{Approximation Algorithm for \MPBST{}}
\label{approx-section}
{\color{mycolor}In this section we first present a randomized approximation algorithm for the \MPBST{} problem. Then we show how to derandomize the algorithm at the expense of increasing the running time by a quadratic factor. The  following theorem summarizes our result in this section. \revised{Throughout this section we consider point sets in the plane that are in general position, that is, no three points lie on a straight line.}

\revised{
\begin{theorem}
There is a randomized algorithm that, given a set of $n$ red and blue points in the plane in general position, returns a plane bichromatic spanning tree of expected weight at most $O(\log n)$ times the optimum, and runs in $O(n\log^2 n)$ time. The algorithm can be derandomized by increasing the running time by a factor of $O(n^2)$.
\end{theorem}}
}

Let $S$ be a set of $n$ red and blue points in the plane. To simplify our arguments we assume that $n$ is a power of $2$. Let $\OPT{}$ denote the length of a minimum bichromatic spanning tree on $S$ (and note that $\OPT$ is an obvious lower bound for the length of a minimum \emph{plane} bichromatic spanning tree on $S$). Our algorithm computes a plane bichromatic spanning tree of \revised{expected} length $O(\log n)\cdot\OPT$.

\subsection{Preliminaries for the algorithm}

The following folklore lemma, though very simple, plays an important role in our construction. 

\begin{lemma}
\label{exsitence-lemma}
    Every set of $n$ red and blue points  \revised{in general position} in the plane, containing at least one red and at least one blue point, admits a plane bichromatic spanning tree. Such a tree can be computed in $O(n\log n)$ time. 
\end{lemma}
{\color{mycolor}A proof of Lemma~\ref{exsitence-lemma} can be found in \cite{Biniaz2019}. Essentially such a tree can be constructed by connecting an arbitrary red point to all blue points (this partitions the plane into cones) and then connecting red points in each cone to a blue point on its boundary.}

For a connected geometric graph $G$ and a point $q$ in the plane, we say that $q$ {\em sees} an edge $(a,b)$ of $G$ if 
the interior of the triangle $\bigtriangleup qab$ is disjoint from  vertices and edges of $G$.
In other words, the entirity of the edge $(a,b)$ is visible from $q$. The following lemma (that is implied from  \cite[Lemma 2.1]{HurtadoKRT08}) also plays an important role in our construction. 

\begin{lemma}
\label{merge-point-cor}
    Let $G$ be a connected plane geometric graph \revised{with $n$ vertices} and $q$ be a point outside the convex hull of the vertices of $G$. Then $q$ sees an edge of $G$. \revised{Such an edge can be found in $O(n\log n)$ time.}
\end{lemma}
Note that the condition that $q$ lies outside of the convex hull of $G$ is necessary, as otherwise $q$ may not see any edge of $G$ entirely. The application of this lemma to our algorithm is that if $G$ is properly colored and a vertex sees an edge $(a,b)$, then $(q,a)$ or $(q,b)$ is bichromatic and does not cross any edges of $G$. This idea was previously used in~\cite{HoffmannST10,Hoffmann2014,HurtadoKRT08}.

\subsection{The algorithm}
\label{algorithm-section}
After a suitable scaling, we may assume that the smallest axis-aligned square containing $S$  has side length $1$. After a suitable translation, we may assume that the lower left corner of this square is the point $(1,1)$; and so its top right corner is $(2,2)$ as in Figure~\ref{quadtree-fig}(a). Observe that \[\OPT{}\ge 1.\] 

Our algorithm employs a randomly shifted quadtree as in Arora’s PTAS for the Euclidean TSP~\cite{Arora1998}.  Let $Q$ be a $2\times 2$ axis-aligned square whose lower left corner is the origin---$Q$ contains all points of $S$. Subdivide $Q$ into four congruent squares, and recurse until $Q$ is subdivided into squares of side length $\nicefrac{1}{n}$ as in Figure~\ref{quadtree-fig}(a). The depth of this recursion is $1+\log n$. For the purpose of shifting, pick two real numbers $x$ and $y$ in the interval $[0,1]$ uniformly at random. Then translate $Q$ such that its lower left corner becomes $(x,y)$ as in Figure~\ref{quadtree-fig}(b). This process is called a {\em random shift}. The points of $S$ remain in $Q$ after the shift. We obtain a quadtree subdivision of $S$ of depth at most $1+\log n$ with respect to the subdivision of $Q$, i.e., the lines of the quadtree subdivision are chosen from the lines of the subdivision of $Q$; see Figure~\ref{quadtree-fig}(c). 
The resulting quadtree is called (randomly) {\em shifted quadtree}. We stop the recursive subdivision at squares that have size $\nicefrac{1}{n}\times\nicefrac{1}{n}$ or that are \emph{empty} (disjoint from $S$) or {\em monochromatic} (have points of only one color). Therefore a leaf-square of the quadtree may contain more than one point of $S$.

\begin{figure}[ht!]
	\centering
\setlength{\tabcolsep}{0in}
	$\begin{tabular}{ccc}
		\multicolumn{1}{m{.33\columnwidth}}{\centering\vspace{0pt}\includegraphics[width=.33\columnwidth]{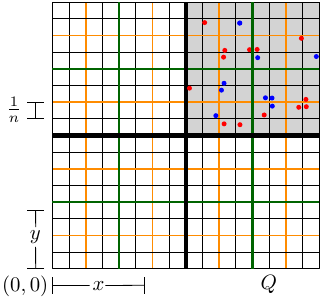}}
		&\multicolumn{1}{m{.4\columnwidth}}{\centering\vspace{10pt}\includegraphics[width=.39\columnwidth]{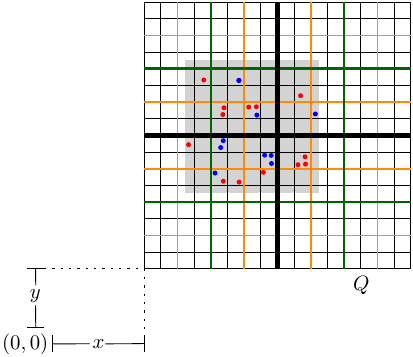}}
		&\multicolumn{1}{m{.27\columnwidth}}{\centering\vspace{-26pt}\includegraphics[width=.255\columnwidth]{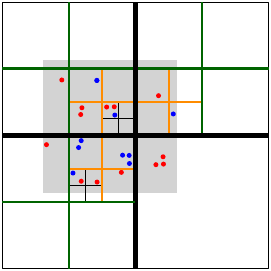}}
		\\
		(a)   &(b) &(c)
	\end{tabular}$	
\caption{(a) Shaded square contains $S$. (b) Translated subdivision of $Q$. (c) Randomly shifted quadtree on points of $S$ with respect to the subdivision of $Q$.}
	\label{quadtree-fig}
\end{figure}

At the root level (which we may consider as level $-1$) we have $Q$ which has size $2\times 2$. At each recursive level $i=0,1,\ldots,\log n$ we have squares of size $\nicefrac{1}{2^i} \times \nicefrac{1}{2^i}$ (Level $0$ stands for the first time that we subdivide $Q$). Thus at level $0$ we have four squares of size $1\times 1$, and at level $\log n$ we have squares of size $\nicefrac{1}{n}\times\nicefrac{1}{n}$.
Our strategy is to use the shifted quadtree and compute an approximate solution in a bottom-up fashion from the leaves towards the root. For each square that is {\em bichromatic} (contains points of both colors) we will find a plane bichromatic spanning tree of its points. For monochromatic squares, we do not do anything.
At the root level, we have $Q$ which contains $S$ and is bichromatic, so we will get an approximate plane bichromatic spanning tree of $S$.

At level $i=\log n$, we have squares of size $\nicefrac1n \times \nicefrac1n$, and thus the length of any edge in such squares is at most $\nicefrac{\sqrt{2}}{n}$. For each bichromatic square we compute a plane bichromatic tree arbitrarily, for instance by Lemma~\ref{exsitence-lemma}. For monochromatic squares we do nothing. The spanning trees for all the squares have less than $n$ edges in total. Hence the total length of all these trees is at most \[\frac{\sqrt{2}}{n}\cdot n=\sqrt{2}\le \sqrt{2}\cdot\OPT{}.\]

\begin{figure}[ht!]
	\centering
\setlength{\tabcolsep}{0in}
	$\begin{tabular}{ccc}
		\multicolumn{1}{m{.33\columnwidth}}{\centering\vspace{0pt}\includegraphics[width=.25\columnwidth]{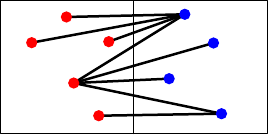}}
		&\multicolumn{1}{m{.33\columnwidth}}{\centering\vspace{0pt}\includegraphics[width=.25\columnwidth]{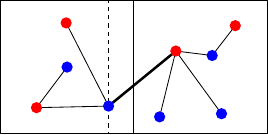}}
		&\multicolumn{1}{m{.33\columnwidth}}{\centering\vspace{0pt}\includegraphics[width=.25\columnwidth]{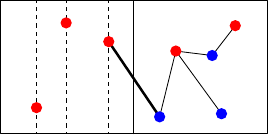}}
		\\
		(a)   &(b) &(c)
	\end{tabular}$	
	\caption{Merging two squares in the same row: (a) case 1, (b) case 2, and (c) case 3.}
	\label{merge-fig}
\end{figure}

At each level $i<\log n$, we have already solved the problem within squares of level $i+1$. Each square at level $i$ has size $\nicefrac{1}{2^i}\times \nicefrac{1}{2^i}$. If the square is monochromatic we do nothing. If the square is bichromatic, then it consists of four squares at level $i+1$. 
We merge the solutions of the four squares to obtain a solution for the level~$i$ square as follows. First we merge the solutions in adjacent squares in the same row (we have two such pairs), and then merge the two solutions in the two rows. Thus each merge is performed on two solutions that are separated by a (vertical or horizontal) line. For the merge we apply one of the following cases:
\begin{itemize}[leftmargin=18pt]
     \item[(1)] If each merge party is monochromatic but their union is bichromatic, then we construct an arbitrary plane bichromatic tree on the union, for example by Lemma~\ref{exsitence-lemma}. See Figure~\ref{merge-fig}(a).
    \item[(2)] If both merge parties are bichromatic (and hence are trees), then we take a point in one tree that is closest to the square (or rectangle) containing the other tree and merge them using Lemma~\ref{merge-point-cor}. See Figure~\ref{merge-fig}(b).
    \item[(3)] If one part is bichromatic (a tree) and the other is monochromatic, then we first sort the points in the monochromatic square (or rectangle) in increasing order according to their distance to the bichromatic square (or rectangle). Then we merge the points (one at a time) with the current bichromatic tree by using Lemma~\ref{merge-point-cor}. See Figure~\ref{merge-fig}(c).
\end{itemize}

In each case, the merge produces a plane bichromatic tree in the level-$i$ square. We process all squares in a bottom-up traversal of the quadtree. In the end, after processing level $-1$, we get a plane bichromatic spanning tree for points of $S$ in square $Q$. Denote this tree by $T$.

\subsection{Weight analysis}
\label{ssec:analysis}

We start by introducing an alternative measurement for the length of the optimal tree{\color{mycolor}, i.e., \MBST{}}. This new measurement, denoted by $\OPT'$, will be used to bound the length of our tree $T$. We say that a quadtree line is at level $i$ if it contains \revised{a side} of some level-$i$ square. \revised{A side} of a level-$i$ square (that is not a leaf) gets subdivided to yield \revised{sides} of two
squares at level $i+1$. Thus any quadtree line at level $i$ is also at levels  $i+1,i+2,\dots,\log n$.

For each level $i\in\{0,1,\dots,\log n\}$ we define a parameter $\OPT'_i$. 
Let $T_{\rm opt}$ be a {\color{mycolor}\MBST{}} for $S$. For each edge $e$ of $T_{\rm opt}$,
define the indicator variable 
\[
X_i(e) =\left\{ \begin{array}{ll} 
1 & \mbox{\rm if $e$ intersects the boundary of a level-$i$ square of the quadtree,} \\
0& \mbox{\rm otherwise.}
\end{array}\right.
\]
Now let \[\OPT_i'=\sum_{e\in E(T_{\rm opt})}  \frac{1}{2^i}X_i(e),\] that is, a weighted sum of $X_i(e)$ over all edges of $T_{\rm opt}$, and let \[\OPT'=\sum_{i=0}^{\log n} \OPT_i'.\]


The new measure $\OPT'$ can be arbitrarily large compared to $\OPT$ in the worst case. For example if the optimal solution is a path consisting of $n-1$ edges of small length say $2/n$, such that each of them intersects the vertical line at level $0$, then $\OPT$ is roughly $2$ but $\OPT'$ is at least $n-1$. However, using the random shift at the beginning of our algorithm, we can show that the expected value of $\OPT'$ is not very large compared to $\OPT$. 

\begin{lemma}
\label{opt-lemma-logn}
    $\EX[\OPT']\leq \sqrt{2}(1+\log n) \cdot\OPT.$
\end{lemma}
\begin{proof}
Consider any edge $e$ of length $\ell$ in {\color{mycolor}\MBST}. Let $\ell_x$ denote the \emph{$x$-span} of $e$, i.e., the difference of the $x$-coordinates of the two endpoints of $e$, and let $\ell_y$ be the \emph{$y$-span} of $e$. Observe that $\ell_x+\ell_y\le \sqrt{2}\ell$.
The probability that $e$ intersects a vertical line at any level $i$ is
\[  \PR(e\text{ intersects a vertical line at level $i$}) \le \ell_x 2^i,
\]
because there are $2^i$ uniformly spaced but randomly shifted vertical lines at level $i$.
Combining this with the probability of intersecting horizontal lines, the union bound yields
\[  \PR(X_i(e))=
\PR(e\text{ intersects a line at level $i$}) \le (\ell_x+\ell_y) 2^i\le\sqrt{2}\ell2^{i}.\]

If $e$ intersects a line at level $i$, then $\nicefrac{1}{2^i}$ is added to $\OPT'_i$, otherwise nothing is added. Thus the expected added value for $e$ to $\OPT_i'$ is 
\[ \EX\left[\frac{1}{2^i} X_i(e)\right] 
=\frac{1}{2^i}\cdot \PR(X_i(e)) 
\le \frac{1}{2^i}\cdot \sqrt{2}\ell 2^i=\sqrt{2}\ell.\]
Summation over all edges of the optimal tree $T_{\rm opt}$ gives $\EX[\OPT'_i]\le\sqrt{2}\cdot\OPT$, and summation over all levels yields  $\EX[\OPT']\le\sqrt{2}(1+\log n) \cdot\OPT$.
\end{proof}

We already know that the total length of trees constructed in level $\log n$ is at most $\sqrt{2}\cdot\OPT$. Let $E_i$ be the set of all edges that were added to $T$ at level $i<\log n$ in the bottom-up construction. We  establish a correspondence between the edges in $E_i$ and the values added to $\OPT'$. The edges of $E_i$ were added by cases (1), (2), and (3). We consider each case separately. 
\begin{itemize}[leftmargin=18pt]
    \item[(1)] Assume that the union of merge parties has $k$ points. Then we add $k-1$ edges of length at most $\sqrt{2}\cdot\nicefrac{1}{2^i}$ to $E_i$. For these points, 
    {\color{mycolor}\MBST{}}
    also needs to have at least $k-1$ connections that intersect the boundaries of the squares  involved in the merge, which have side length at least $\nicefrac{1}{2^{i+1}}$. For each such  edge we have added a value of $\nicefrac{1}{2^{i+1}}$ to $\OPT_{i+1}'$. 
    \item[(2)] To merge the two trees, we added just one edge of length at most $\sqrt{2}\cdot\nicefrac{1}{2^i}$ to $E_i$. 
    {\color{mycolor} For each merge party, \MBST{} needs at least one edge that crosses the boundary of its rectangle. These edges, however, need not be distinct (e.g., an edge can cross the boundary between the two rectangles). 
    In any case, \MBST{} needs at least one edge that crosses the boundary of one of the two  rectangles, for which we have added at least $\nicefrac{1}{2^{i+1}}$ to $\OPT'_{i+1}$.}
    \item[(3)] Assume that the monochromatic party has $k$ points. Thus we added $k$ edges of length at most $\sqrt{2}\cdot\nicefrac{1}{2^i}$ to $E_i$. Again, 
    {\color{mycolor}\MBST{}}
    needs at least $k$ edges that cross the boundary of the squares involved in the merge; and for each such edge we have added at least $\nicefrac{1}{2^{i+1}}$ to $\OPT'_{i+1}$.
\end{itemize}

Therefore the total length of all edges that were added to $T$ at level $i$ is at most $\revised{|E_i|}\cdot \nicefrac{\sqrt{2}}{2^i}$.
Analogously, the total value that has been added to $\OPT'_{i+1}$ is at least $\nicefrac{1}{2}\cdot |E_i|\cdot \nicefrac{1}{2^{i+1}}=|E_i|\cdot \nicefrac{1}{2^{i+2}}$. The multiplicative factor $\nicefrac{1}{2}$ comes from the fact that the two endpoints of an edge of the optimal tree could be involved in two separate merge operations. By summing the length of edges added to $T$ in all levels and considering Lemma~\ref{opt-lemma-logn} we get

\begin{align*}
\notag \EX[|T|]&\le\sqrt{2}\cdot\OPT+\EX\left[\sum_{i=0}^{\log n -1}{|E_i|\frac{\sqrt{2}}{2^i}}\right]
\le \sqrt{2}\cdot\OPT+4\sqrt{2}\cdot\EX\left[\sum_{i=0}^{\log n -1}{|E_i|\frac{1}{2^{i+2}}}\right]\\ \notag
&\le  \sqrt{2}\cdot\OPT+4\sqrt{2}\cdot\EX\left[\sum_{i=0}^{\log n -1}{\OPT'_{i+1}}\right]\le  \sqrt{2}\cdot\OPT+4\sqrt{2}\cdot\EX[\OPT']\\ \notag
&\le  \sqrt{2}\cdot\OPT+4\sqrt{2}\cdot \sqrt{2}(1+\log n) \cdot\OPT =O(\log n)\cdot \OPT.
\end{align*}


\subsection{Derandomization}
In the algorithm of Section~\ref{algorithm-section}, we shifted the $2\times 2$ square $Q$ by a real vector $(x,y)$, where $x$ and $y$ are chosen independently and uniformly at random from the interval $[0,1]$. We now \emph{discretize} the random shift, and choose $x$ and $y$ independently and uniformly at random from the finite set $\left\{0,\nicefrac{1}{n},\nicefrac{2}{n},\ldots , \nicefrac{n-1}{n}\right\}$. We call this process the \emph{discrete  random shift}. We show that the proof of Lemma~\ref{opt-lemma-logn} can be adapted under this random experiment with a larger constant coefficient. {\color{mycolor}Therefore  we can derandomize the algorithm by trying all random choices of $(x,y)$ for the shift and return the shortest tree over all choices. This increases the running time by a factor of $O(n^2)$.}

\begin{lemma}
Under the discrete random shift, we have
\label{opt-lemma-discrete}
    $\EX[\OPT']\leq (\sqrt{2}+2)(1+\log n) \cdot\OPT$.
\end{lemma}
\begin{proof}
Consider any edge $e=ab$ of length $\ell$ in the optimal tree $T_{\rm opt}$. The two endpoints of $e$ are points $a=(a_x,a_y)$ and $b=(b_x,b_y)$. Denote the orthogonal projection of $e$ to the $x$- and $y$-axes by $e_x$ and $e_y$, respectively, and observe that $e_x=[\min\{a_x,b_x\},\max\{a_x,b_x\}]$ and $e_y=[\min\{a_y,b_y\},\max\{a_y,b_y\}]$. 
We discretize these intervals as follows. Replace each \revised{endpoint} of the interval $e_x$ with the closest rationals of the form $\nicefrac{k}{n}+\nicefrac{1}{2n}$, and let $e'_x$ be the resulting interval; similarly we obtain $e'
_y$ from $e_y$. \revised{Observe that $e_x$ intersects a vertical line of the form $x=\nicefrac{a}{n}$, $a\in \mathbb{Z}$, if and only if $e_x'$ does.} Let $\ell_x'$ and $\ell_y'$ denote the lengths of \revised{ intervals $e'_x$ and $e'_y$, respectively.} By construction, we have $\ell'_x\leq \ell_x+\nicefrac{1}{n}$ and $\ell'_y\leq \ell_y+\nicefrac{1}{n}$.  Consequently, the probability that $e$ intersects a vertical line at any level $i$ is  
\begin{align*} 
\PR(e\text{ intersects a vertical line at level $i$}) 
& = \PR(e'_x\text{ intersects a vertical line at level $i$}) \\
&\le \ell'_x 2^i \leq \left(\ell_x+\nicefrac{1}{n}\right) 2^i,
\end{align*}
because there are $2^i$ uniformly spaced but randomly shifted vertical lines at level $i$. Combining this with the probability of intersecting horizontal lines, the union bound yields
\[  \PR(X_i(e))=
\PR(e\text{ intersects a line at level $i$}) 
\le \left(\ell_x+\ell_y+\nicefrac{2}{n}\right) 2^i
\le \left(\sqrt{2}\ell + \nicefrac{2}{n}\right) 2^i.
\]
The expected added value for $e$ to $\OPT_i'$ is 
\[\EX\left[\frac{1}{2^i} X_i(e)\right] 
=\frac{1}{2^i}\cdot \PR(X_i(e)) 
\le \frac{1}{2^i}\cdot \left(\sqrt{2}\ell + \frac{2}{n}\right)2^i
=\sqrt{2}\ell + \frac{2}{n}.\]
Summation over all edges of the optimal tree $T_{\rm opt}$ gives $\EX[\OPT'_i]\le\sqrt{2}\cdot\OPT + 2 \leq (\sqrt{2}+2)\cdot \OPT$, and summation over all levels yields  $\EX[\OPT']\le(\sqrt{2}+2)(1+\log n) \cdot\OPT$.
\end{proof}

{\color{mycolor}
The initial shifted quadtree has depth $O(\log n)$ and has $O(n)$ leaves. Thus it can be computed in $O(n\log n)$ time by a divide-and-conquer sorting-based algorithm~\cite{deBerg2008}. \revised{From a result of \cite{Biniaz2019} (cf.~Lemma~\ref{exsitence-lemma}) it follows that the bichromatic trees at the leaves of the quadtree can be computed in total $O(n\log n)$ time. From a result of \cite{HershbergerS92} and \cite{HurtadoKRT08} (cf.~Lemma~\ref{merge-point-cor})} it follows that the total merge time in each level of the quadtree is $O(n\log n)$. Summing over all levels, the running time of our algorithm is $O(n\log^2 n)$. 
}

\subsection{Generalization to more colors}
{\color{mycolor}
Our approximation algorithm for the \MPBST{} (minimum plane two-colored spanning tree) can be generalized to more colors. In this general setting, we are given a colorful point set $S$ and we want to find a spanning tree on $S$ with properly colored edges, i.e., the two endpoints of every edge should be of different colors. The same quadtree approach would give a plane properly-colored spanning tree. The analysis and the approximation ratio would be the same mainly because whenever we introduce some edges to merge two squares, the optimal solution must have the same number of edges that cross the boundaries of the squares. Also the running time remains the same because Lemma~\ref{exsitence-lemma} and Lemma~\ref{merge-point-cor} carry over to multicolored point sets.
}


{\color{mycolor}
\section{Crossing Patterns in \MBST{}}
\label{com-sec}

In this section, we prove that \MBST{} is quasi-plane for every 2-colored point set in general position (Section~\ref{quasi-section}), and then use this \revised{result} to determine the maximum number of crossings in \MBST{} for a set of $n$ bichromatic points in the plane (Section~\ref{crossings-section}).}

\subsection{Quasi-planarity}
\label{quasi-section}
Let $S$ be a set of red and blue points in the plane. To differentiate between the points we denote the red points by $r_1,r_2,\dots$ and the blue points by $b_1,b_2,\dots$. Let $T$ be a MinBST for $S$. For two distinct edges $e_1$ and $e_2$ of $T$ we denote the unique shortest path between $e_1$ and $e_2$ in $T$ by $\spath{e_1}{e_2}$. This path contains exactly one endpoint of $e_1$ and  one endpoint of $e_2$.    

\begin{figure}[ht!]
	\centering
\includegraphics[width=.27\columnwidth]{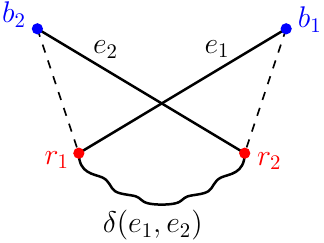}	
	\caption{Illustration of the proof of Lemma~\ref{endpoint-lemma}. Uncrossing a pair of crossing edges.}
	\label{uncross-fig}
\end{figure}

\begin{lemma}
\label{endpoint-lemma}
Let $e_1$ and $e_2$ be two edges of $T$ that cross each other. Then the endpoints of $\spath{e_1}{e_2}$ have different colors. 
\end{lemma}
\begin{proof}
Let $e_1=(r_1,b_1)$ and $e_2=(r_2,b_2)$.
Suppose, for the sake of contradiction, that the endpoints of $\spath{e_1}{e_2}$ are of the same color, w.l.o.g.\ red. Then the endpoints of $\spath{e_1}{e_2}$  are $r_1$ and $r_2$ as in Figure~\ref{uncross-fig}. In this case, we can replace edges $(r_1,b_1)$ and $(r_2,b_2)$ of $T$ by two new edges $(r_1,b_2)$ and $(r_2,b_1)$ and obtain a new bichromatic spanning tree $T'$. By the triangle inequality \revised{(applied to each of the two triangles induced by the crossing)}, the total length of the two new edges is smaller than the total length of the two orginal edges. Hence $T'$ is shorter \revised{than} $T$, contradicting the minimality of $T$.  
\end{proof}

\begin{figure}[ht!]
	\centering
	\setlength{\tabcolsep}{0in}
	$\begin{tabular}{cc}
		\multicolumn{1}{m{.5\columnwidth}}{\centering\vspace{0pt}\includegraphics[width=.37\columnwidth]{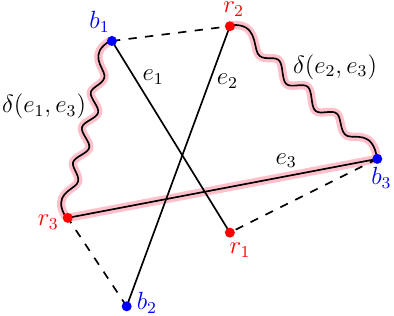}}
		&\multicolumn{1}{m{.5\columnwidth}}{\centering\vspace{0pt}\includegraphics[width=.32\columnwidth]{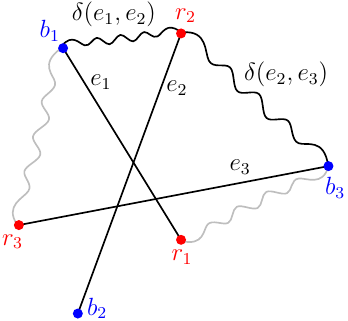}}
		\\
		(a)   &(b) 
	\end{tabular}$	
	\caption{(a) Replacing $e_1,e_2,e_3$ by $(r_1,b_3), (r_3,b_2),(r_2,b_1)$; the highlighted path is $\spath{e_1}{e_2}$. (b) Getting a cycle in the union of $\spath{e_1}{e_2}$, $\spath{e_2}{e_3}$, $\spath{e_1}{e_3}$, \revised{together} with $e_1$ or $e_3$. \revised{Gray paths represent the two possible choices for $\spath{e_1}{e_3}$.}}
	\label{three-crossing-fig}
\end{figure}

\begin{theorem}
\label{three-planar}
\revised{Every} Euclidean minimum bichromatic spanning tree is quasi-plane.
\end{theorem}
\begin{proof}
\revised{Let $T$ be a Euclidean minimum bichromatic spanning tree. We prove that no three edges of $T$ can pairwise cross each other. This will imply that $T$ is quasi-plane.}
The proof proceeds by contradiction. Suppose that three edges of $T$, say $e_1=(r_1,b_1)$, $e_2=(r_2,b_2)$ and $e_3=(r_3,b_3)$, pairwise cross each other as in Figure~\ref{three-crossing-fig}. We consider the following two cases:

\begin{itemize}
    \item[1.] $\spath{e_i}{e_j}$ contains $e_k$ for 
    some permutation of the indices with $\{i,j,k\}=\{1,2,3\}$. 

    After a suitable relabeling assume that $\spath{e_1}{e_2}$ contains $e_3$. Then $\spath{e_1}{e_3}$ and $\spath{e_2}{e_3}$ are sub-paths of $\spath{e_1}{e_2}$ and they do not contain $e_2$ and $e_1$ respectively. Assume without loss of generality (and by Lemma~\ref{endpoint-lemma}) that the endpoints of $\spath{e_1}{e_2}$ are $r_2$ and $b_1$ as in Figure~\ref{three-crossing-fig}(a); $\spath{e_1}{e_2}$ is highlighted in the figure. For the rest of our argument we will use a result of \cite{BiniazMS2021} that for an even number of (monochromatic) points in the plane, a perfect matching with pairwise crossing edges is the unique maximum-weight matching  (without considering the colors). This means that $M=\{e_1,e_2,e_3\}$ is the maximum matching for the point set $\{r_1,r_2,r_3,b_1,b_2,b_3\}$. Therefore $M$ is longer than $R=\{(r_1,b_3),(r_2,b_1),(r_3,b_2)\}$, which is a (bichromatic) matching for the same points. By replacing the edges of $M$ in $T$ with the edges of $R$, we obtain a shorter tree $T'$, contradicting \revised{the minimality} of $T$. To verify that $T'$ is \revised{a tree} imagine replacing the edges one at a time. If we add $(r_1,b_3)$ first, then we create a cycle that contains $e_1$, and thus by removing  $e_1$ we obtain a valid tree. Similarly we can replace $e_2$ by $(r_3,b_2)$ and $e_3$ by $(r_2,b_1)$.

    \item[2.] $\spath{e_i}{e_j}$ does not contain $e_k$ for 
    any permutation of the indices with $\{i,j,k\}=\{1,2,3\}$. 

    Consider the path $\spath{e_1}{e_2}$ and assume without loss of generality (and by Lemma~\ref{endpoint-lemma}) that its endpoints are $r_2$ and $b_1$ as in Figure~\ref{three-crossing-fig}(b). Now consider $\spath{e_2}{e_3}$.  This path cannot have $r_3$ and $b_2$ as its endpoints because otherwise the path $\spath{e_1}{e_3}$ would contain $e_2$, contradicting the assumption of the current case. Therefore  the endpoints of $\spath{e_2}{e_3}$ are $r_2$ and $b_3$. Now consider the path $\spath{e_1}{e_3}$. If its endpoints are $r_1$ and $b_3$, then \revised{the union of $\spath{e_1}{e_2}$, $\spath{e_2}{e_3}$, and $\spath{e_1}{e_3}$ contains a path between $r_1$ and $b_1$ that does not go through $e_1$; the union of this path and $e_1$ is a cycle in $T$}. Similarly, if \revised{the} endpoints of $\spath{e_1}{e_3}$ are $r_3$ and $b_1$, then the union of $\spath{e_1}{e_2}$, $\spath{e_2}{e_3}$, $\spath{e_1}{e_3}$, and $e_3$ contains a cycle. Both cases lead to a contradiction as $T$ has no cycle.\qedhere
\end{itemize}
\end{proof}

{\color{mycolor}
\subsection{Maximum number of crossings}
\label{crossings-section}
Given that \revised{a \MBST{}} is quasi-plane (Theorem~\ref{three-planar}), one wonders how many crossings it can have. As illustrated in Figure~\ref{MinBST-fig}, the number of crossings per edge can be linear in the number of points, and the total number of crossings can be quadratic. We give tight upper bounds for both quantities (Propositions~\ref{pp:crossings}--\ref{pp:localcrossings}), and also show that \MBST{} always has a crossing-free edge (Proposition~\ref{pp:shortest-edge}).

\begin{proposition}\label{pp:shortest-edge}
For every finite set of bichromatic points in the plane in general position, every \MBST{} contains a closest bichromatic pair as an edge. Moreover, no such edge is intersected by other edges of the \MBST{}. 
\end{proposition}
\begin{proof}
We prove both parts by contradiction. For the first part assume that the \MBST{} does not contain an edge between any closest bichromatic pair. 
Let $\{r_1,b_1\}$ be a closest bichromatic pair. By adding $(r_1,b_1)$ to the tree we obtain a cycle in which $(r_1,b_1)$ is the shortest edge. By removing any other edge from the cycle we obtain a bichromatic spanning tree of a shorter length. This contradicts the minimality of the original \MBST{}. 

For the second part let $\{r_1,b_1\}$ be a closest bichromatic pair that appears as an edge in the \MBST{}.  \revised{Hence} $e_1=(r_1,b_1)$ crosses some edge $e_2=(r_2,b_2)$. Without loss of generality (and by Lemma~\ref{endpoint-lemma}), assume that the endpoints of $\spath{e_1}{e_2}$ are $r_1$ and $b_2$. If $|r_2b_1|<|r_2b_2|$, then by replacing $(r_2,b_2)$ with $(r_2,b_1)$ we obtain a shorter bichromatic spanning tree, a contradiction. Assume that $|r_2b_1|\ge|r_2b_2|$. Since $\{r_1,b_1\}$ is a closest bichromatic pair we have $|r_1b_2|\ge |r_1b_1|$. Adding the two inequalities yields $|r_1b_2|+|r_2b_1|\ge|r_1b_1|+|r_2b_2|$. Note, however, that the vertices of any two crossing edges form a convex quadrilateral. By the triangle inequality, the total length of the two diagonals of a convex quadrilateral is strictly \revised{more} than the total length of any pair of opposite edges, yielding $|r_1b_2|+|r_2b_1|<|r_1b_1|+|r_2b_2|$, a contradiction.
\end{proof}
}

{\color{mycolor}
\begin{proposition}\label{pp:crossings}
For every set of \revised{$n\geq 2$} bichromatic points in the plane in general position, \revised{every} \MBST{} has at most 
$\revised{\lfloor\nicefrac{n^2}{4}\rfloor}-n+1$ 
crossings, and this bound is the best possible. 
\end{proposition}
\begin{proof}
\revised{To verify that the claimed bound can be attained},
consider the construction in Figure~\ref{MinBST-fig} where \revised{the two top clusters have $\lfloor n/2\rfloor{-}1$ and $\lceil n/2\rceil{-}1$ points. Then the total number of crossings in $\text{\MBST}(S)$ is $(\lfloor n/2\rfloor{-}1)\cdot(\lceil n/2\rceil{-}1)$,
which is equal to $\revised{\lfloor\nicefrac{n^2}{4}\rfloor}-n+1$ because $n$ is an integer.}

For an upper bound, let $S$ be a set of $n$ bichromatic points in general position. Define the \emph{crossing graph} $G_{\rm cr}$ of $\text{\MBST}(S)$, where the vertices of $G_{\rm cr}$ correspond to the edges of $\text{\MBST}(S)$ and edges of $G_{\rm cr}$ represent crossings between the edges of $\text{\MBST}(S)$. Note that $G_{\rm cr}$ has $n-1$ vertices where one of them is of degree 0 by Proposition~\ref{pp:shortest-edge}. By Theorem~\ref{three-planar}, $G_{\rm cr}$ is triangle-free. Therefore, by Tur\'an's theorem~\cite{Turan41}, $G_{\rm cr}$ has at most 
\revised{$(\lfloor n/2\rfloor{-}1)\cdot(\lceil n/2\rceil{-}1)$}
edges. Consequently, $\text{\MBST}(S)$ has at most this many crossings.
\end{proof}

\begin{proposition}\label{pp:localcrossings}
For every set of \revised{$n\geq 3$} bichromatic points in the plane in general position, \revised{every} edge of a \MBST{} crosses at most $n-3$ other edges, and this bound is the best possible.
\end{proposition}
\begin{proof}
\revised{To verify that the claimed bound can be attained, consider the construction in Figure~\ref{MinBST-fig} and replace one of the top clusters by a single point and the other by $n-3$ points.} Then all points in this cluster have degree 1 in $\text{\MBST}(S)$, these $n-3$ leaves all cross one edge of \MBST{}. 

For the upper bound, notice that a \MBST{} has $n-1$ edges, one of which is crossing-free by Proposition~\ref{pp:shortest-edge}. Consequently, an edge in \revised{a \MBST{}} can cross at most $n-3$ other edges. 
\end{proof}
}

\section{Conclusions and Open Problems}\label{sec:conclusion}

{\color{mycolor}
We conclude with a collection of open problems raised by our results. 
We have presented a $O(\log n)$-approximation algorithm for the \MPBST{} problem for a set of $n$ bichromatic points in the plane, and showed that $\rho_n\leq O(\log n)$. Recall that the current best lower bound is $\rho_n\geq 3/2$ for all $n\geq 4$~\cite{Grantson2005}. It remains open whether a constant-factor approximation is possible, whether the problem is APX-hard, and whether $\rho_n$ is bounded by a constant. 

It is also natural to investigate whether there is an (approximation) algorithm that, given a bichromatic point set and an integer $d$, finds a \emph{minimum plane bichromatic tree of maximum degree at most $d$} (or reports that none exists). It is known that any set of $n$ red and $n$ blue points in general position admits a plane bichromatic spanning tree of maximum degree
at most three~\cite{Kaneko1998}; but there are $n$ red and $n$ blue points in convex position that do not admit a bichromatic plane spanning path~\cite{AkiyamaU90a}. For the general case of $n$ red and $m$ blue points, with $n\ge m$, there exists a plane bichromatic spanning tree of maximum degree at most $\max\{3,\lceil\nicefrac{n-1}{m}\rceil+1\}$ and this is the best upper bound \cite{Biniaz2018a}.
}

We have shown that \MBST{} is quasi-plane, which means that the crossing graph $G_{\rm cr}$ of \MBST{} is triangle-free. Figure~\ref{MinBST-fig} shows that  $G_{\rm cr}$ can have 4-cycles (and even cycles of any lengths). 
Can the crossing graph $G_{\rm cr}$ of \MBST{} contain an odd cycle (e.g., a 5-cycle)? 
Can every \MBST{} be decomposed into a constant number of planar straight-line graphs? 
%

\section{Acknowledgement}
This work was initiated at the Eleventh Annual Workshop on Geometry and Graphs, held at the Bellairs Research Institute in Holetown, Barbados in March 2024. The authors thank the organizers and the participants.

\bibliographystyle{plainurl}

\bibliography{MPBST.bib}
\end{document}